\documentclass[twocolumn,5p,number]{elsarticle} 
        
\usepackage{natbib}
\usepackage{graphicx,color,array,hyperref,url}
\usepackage{amsmath,amsthm,amssymb,amsfonts,booktabs,multirow,bm,balance,bbm}
\usepackage{caption}
\usepackage{subcaption}
\usepackage[normalem]{ulem}
\usepackage{algorithm}
\usepackage{algpseudocode}

\usepackage{pgfplots}
\usetikzlibrary{plotmarks}
\usetikzlibrary{arrows.meta}
\usepgfplotslibrary{patchplots}
\usepackage{grffile}
\usepackage{tikz}
\usepackage{cancel}
\usepackage{comment}

\theoremstyle{plain}
\newtheorem{theorem}{Theorem}

\newtheorem{proposition}{Proposition}

\newtheorem*{claim*}{Claim}

\newtheorem*{example*}{Example}
\newtheorem{example}{Example}

\theoremstyle{definition}
\newtheorem{definition}{Definition}

\newtheorem{assumption}{Assumption}
\theoremstyle{remark}
\newtheorem{remark}{Remark}
\newtheorem*{remark*}{Remark}

\newcommand{\R}{\mathbb{R}}

\usepackage{soul}

\begin{document}
\begin{frontmatter}

\title{
Safe-by-Design Planner-Tracker Synthesis
}
\author{Katherine S. Schweidel}
\author{He Yin}
\author{Stanley W. Smith}
\author{Murat Arcak}

\tnotetext[fninfo]{E-mail addresses: \{kschweidel, he\_yin, arcak\}@berkeley.edu}

\begin{abstract}
We present a safe-by-design trajectory planning and tracking framework for nonlinear dynamical systems using a hierarchy of system models. The planning layer uses a low-fidelity model to plan a feasible trajectory satisfying the {planning constraints}, and the tracking layer utilizes the high-fidelity model to design a controller that restricts the error states between the low- and high-fidelity models to a bounded set. The low-fidelity model enables the planning to be performed online (e.g. using Model Predictive Control) and the tracking controller and error bound are derived offline (e.g. using sum-of-squares programming). To provide freedom in the choice of the low-fidelity model, we allow the tracking error to depend on both the states and inputs of the planner. The goal of this article is to provide a tutorial review of this hierarchical framework and to illustrate it with examples, including a design for vehicle obstacle avoidance.

\emph{Keywords: Motion Planning, Hierarchical Control, Sum-of-Squares Programming, Model Predictive Control.}

\end{abstract}

\end{frontmatter}

\section{Introduction}\label{sec:introduction}

{
Modern engineering systems such as autonomous vehicles and UAVs must operate subject to complex safety and performance requirements
in changing environments.
Designing controllers that meet
such requirements in real-time may be computationally intractable, e.g., due to large system dimension or nonlinearities in a high-fidelity dynamical model
of the system.
The planner-tracker framework \cite{tedrake2010lqr, herbert2017fastrack, summetCT, kousik2018bridging, Singh2018RobustTW, YinStan2019, rosolia_mr} addresses this challenge with a layered architecture 
where a lower-fidelity ``planning" model is employed for online planning and a
``tracking" controller, synthesized offline, keeps the 
tracking error between the 
high-fidelity (``tracking") model and the planning model
within a bounded set. 
System safety is then guaranteed if
 the planner constraints, when augmented by the tracking error bound, 
lie within the 
safety constraints.

There is a choice to be made when defining the tracking error between the planner and tracker systems. In \cite{Singh2018RobustTW,YinStan2019,YinMoyACC}, the tracking error depends on only the planner/tracker \textit{states}. In \cite{meyer2019continuous}, the tracking error is generalized to also depend on the planner \textit{input} which allows for kinematic planning models. 
This is achieved by accounting for the jumps in the error variable that are induced by jumps in the {zero-order hold} input between time-steps.
References \cite{YinStan2019} and \cite{meyer2019continuous} further make a connection between the layered planner-tracker architecture and the notion of abstractions introduced in \cite{GirardPappas2009}. In doing so, they also eliminate the restrictive geometric conditions in \cite{GirardPappas2009}, 
also implicit  in
\cite{Singh2018RobustTW},
which require that
the set where the tracking error vanishes be invariant. Removing this requirement and allowing the tracking error to depend on planner inputs greatly expand the applicability of the planner-tracker framework.

In this tutorial we introduce a broad framework which encompasses those earlier results while further generalizing the error definition compared to~\cite{meyer2019continuous}. In addition, the framework described here is not restricted to a particular planner. Indeed, unlike the computationally heavy symbolic design method used for planning
 in~\cite{meyer2019continuous},
 the numerical example presented here uses the popular choice of Model Predictive Control (MPC), which is appropriate for real-time implementation.

Although MPC is often used for both planning and control, under mismatch of planning model and the plant, the MPC optimization problem must be robustified. Feasibility and stability properties of robust MPC have been studied in 
\cite{kothare1996robust, mayne2000constrained} and in subsequent publications. For linear systems, Tube MPC \cite{chisci2001systems,Langson2004RobustMP,Goulart2006, rakovic2012parameterized, rakovic2013homothetic,munoz2013recursively,fleming2014robust,rakovic2016elastic,bujarbaruah2020robust} is a widely used approach that solves a computationally efficient convex optimization problem for robust control synthesis. 
Although Tube MPC design with feasibility and stability properties are proposed for nonlinear systems in \cite{cannon2011robust, allgower2012nonlinear, yu2013tube, kohler2018novel, faulwasser2018economic, kollerFelix1, kohler2019computationally}, the control synthesis problem can become either too conservative, 
or computationally demanding.

The remainder of the paper is organized as follows. Section~\ref{sec:prob_setup} introduces the high-fidelity tracking model and the low-fidelity planning model and defines a simple tracking error that depends only on the planner/tracker \textit{states}. We build intuition with this simple error model and present the method for constructing a tracking controller and an error bound using sum-of-squares (SOS) programming.
Section~\ref{sec:input_jump} generalizes the tracking error definition to additionally depend on the planner \textit{input} and extends the results in Section~\ref{sec:prob_setup} to handle this generalized error. In Section~\ref{examples}, we demonstrate the method on a vehicle obstacle avoidance example, and we provide concluding remarks in Section~\ref{sec:conclusion}.
}



\begin{figure*} [h]
    \centering
    \begin{subfigure}[b]{0.48\textwidth}
    \centering
    \includegraphics[width=0.75\textwidth]{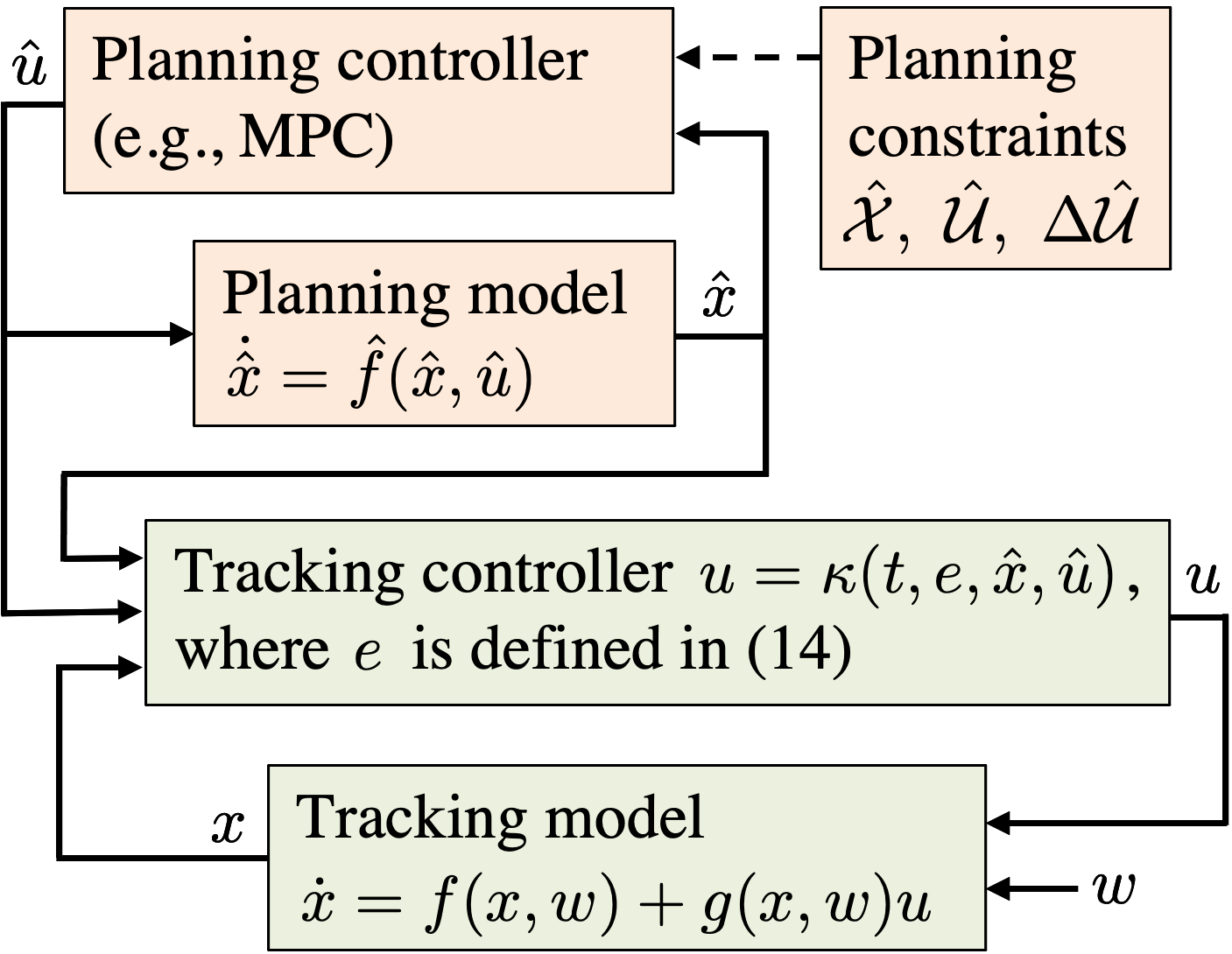}
    \caption{Online implementation. The planning controller uses the state $\hat{x}$ and the constraint sets $\hat{\mathcal{X}}$, $\hat{\mathcal{U}}$, and $\Delta\hat{\mathcal{X}}$ to generate a reference input $\hat{u}$. The tracking controller converts this into a control $u$ which is guaranteed to keep the tracking state $x$ within state constraints $\mathcal{X}$.
    This is accomplished by keeping the tracking error $e$ within a set $\mathcal{O}$ as described in Figure~\ref{fig: diagram}(b).}
    \end{subfigure}
    ~
    \begin{subfigure}[b]{0.48\textwidth}
    \centering
    \includegraphics[width=0.93\textwidth]{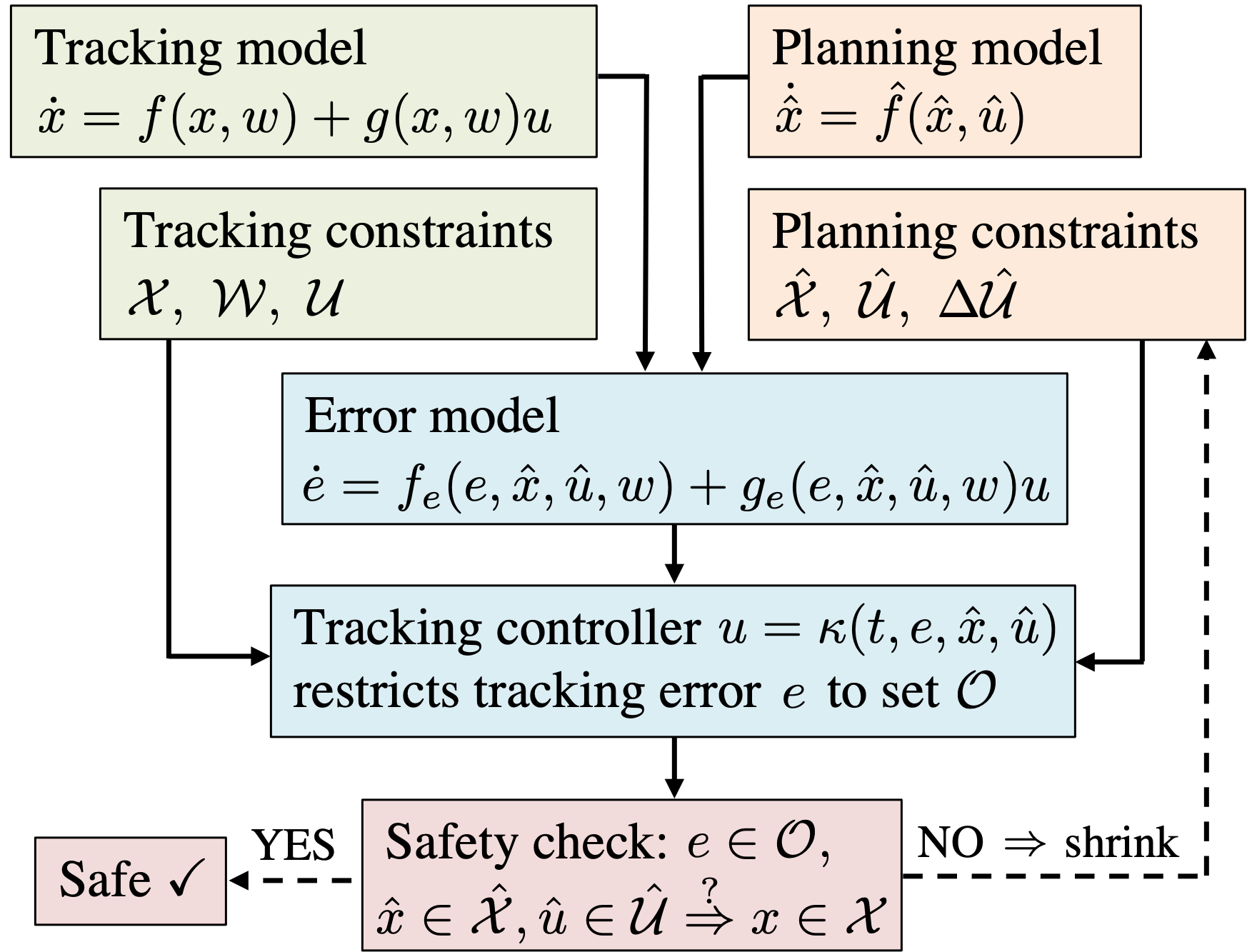}
    \caption{Offline synthesis. The tracking and planning models are combined to obtain a model for the error system. A controller $\kappa$ is derived to keep the error in the set $\mathcal{O}$ (e.g., using SOS). 
    If the safety condition is met,
    the tracking system is guaranteed to satisfy the state constraints $\mathcal{X}$.
    Otherwise, the sets $\hat{\mathcal{X}}$ and $\hat{\mathcal{U}}$ are shrunk and the process is repeated.}
    \end{subfigure}
    \caption{Online implementation and offline synthesis of the planner-tracker control scheme.}
    \label{fig: diagram}
\end{figure*}

\subsection*{Notation}
$\mathbb{S}^n$ denotes the set of $n$-by-$n$
symmetric matrices. $\mathbb{S}_+^n$ and $\mathbb{S}_{++}^n$ denote
the sets of $n$-by-$n$ symmetric positive semi-definite and positive
definite matrices, respectively. For $\xi \in \mathbb{R}^n$, $\mathbb{R}[\xi]$ represents the set of polynomials in $\xi$ with real coefficients, and $\R^{m}[\xi]$ and $\R^{m\times p}[\xi]$ denote all vector and
matrix valued polynomial functions. The subset $\Sigma[\xi] := \{p = p_1^2 + p_2^2 + ... + p_M^2 : p_1, ..., p_M \in \mathbb{R}[\xi]\}$ of $\mathbb{R}[\xi]$ is the set of sum-of-squares polynomials in $\xi$. Unless defined otherwise, notation $x^j$ denotes a variable $x$ used in the $j$'th iteration of an iterative algorithm. The symbol ``$\leq$'' represents component-wise inequality.

\section{Problem setup}\label{sec:prob_setup}
In this section we describe the hierarchical approach to safe-by-design trajectory planning and control that consists of two layers: a planning layer, which uses a low-fidelity system model, and a tracking layer, with a high-fidelity system model. 
The low-fidelity model might be a model with a lower state dimension than the high-fidelity model or a linearized model of the high-fidelity model to reduce the computational burden of planning. By analyzing the dynamics of the error between these two systems' states, we will show how we can bound this error by synthesizing an appropriate tracking controller. 
In this article, the controller and corresponding error bound are designed via sum-of-squares (SOS) programming.

The online implementation and offline synthesis of the planner-tracker control scheme are summarized in Figure~\ref{fig: diagram}.
We begin with a high-fidelity tracking model and a low-fidelity planning model, each with state and input constraints. 
Defining an appropriate error variable, $e$, between the two models, and using the error dynamics and the planner/tracker constraints, we design a tracking controller and derive a tracking error bound. This bound takes the form of a set $\mathcal{O}$ which is invariant under the closed loop error dynamics: $e(0)\in\mathcal{O} \ \Rightarrow \ e(t)\in\mathcal{O}$. 
If the planner constraints, when augmented by $\mathcal{O}$, still satisfy the tracking constraints, then the tracking system is safe: it will satisfy all constraints with the synthesized controller. Otherwise, the planner constraints are shrunk and the process is repeated. 

\subsection{High-Fidelity Tracking Model}
The high-fidelity model used is of the form:
\begin{equation}\label{eq:nonl_system}
\dot{x}(t) = f(x(t), w(t)) + g(x(t), w(t)) \cdot u(t),
\end{equation}
with state $x(t) \in \mathcal{X} \subseteq \R^{n_x}$,  disturbance 
$w(t) \in \mathcal{W} \subseteq \R^{n_w}$, bounded control $u(t) \in \mathcal{U} \subseteq \R^{n_u}$, $f : \R^{n_x} \times \R^{n_w} \to \R^{n_x}$, and $g : \R^{n_x}\times \R^{n_w} \to \R^{n_x} \times \R^{n_u}$. The sets
$\mathcal{X}$ and $\mathcal{U}$ are the constraint sets imposed on the states and control inputs in the high-fidelity model, respectively.

\subsection{Low-Fidelity Planning Model}
The low-fidelity model, which is a simplified version of \eqref{eq:nonl_system}, is of the form:
\begin{equation} \label{lowFidelity}
\dot{\hat{x}}(t) = \hat{f}(\hat{x}(t), \hat{u}(t)),
\end{equation}
where $\hat{x}(t) \in \hat{\mathcal{X}} \subseteq \R^{\hat{n}_x}$, $\hat{u}(t) \in \hat{\mathcal{U}} \subseteq \R^{\hat{n}_u}$, and $\hat{f} : \R^{\hat{n}_x \times \hat{n}_u} \to \R^{\hat{n}_x}$. The sets $\hat{\mathcal{X}}$ and $\hat{\mathcal{U}}$ are constraint sets enforced by the planning layer. The control input for the low-fidelity model, computed via the planning algorithm of choice, is assumed to be a zero-order hold signal with sampling time $T_s > 0$. 
This means: 
\begin{subequations}\label{eq:zero_order_uhat}
\begin{align} 
    &\hat{u}(t) = \hat{u}(\tau_k), \ \forall t \in [\tau_k, \tau_{k+1}), \ \text{with} \ \tau_k = k\cdot T_s, \label{eq:uhat_const} \\
    &\hat{u}(\tau_{k+1}) = \hat{u}(\tau_k) + \Delta\hat{u}(\tau_{k+1}), \label{eq:uhat_jump}
\end{align}
\end{subequations}
where 
$\Delta\hat{u}(t)$ is the periodic change in the control, restricted to a set $\Delta \hat{\mathcal{U}} \subseteq \R^{\hat{n}_u}$.
The zero-order hold behavior of the input will become important in Section \ref{sec:input_jump}, where it will necessitate additional analysis in order to provide a tracking error bound.



\begin{remark}
    Note that the planner-tracker synthesis framework is 
    applicable to any planning algorithm that is able to bound $\hat{x}(t), \hat{u}(t)$, and $\Delta \hat{u}(t)$. For example, this framework has been applied to different planning algorithms, using Nonlinear MPC in \cite{YinMoyACC, YinStan2019}, signal temporal logic (STL) in \cite{pant2020codesign}, and discrete abstraction in \cite{meyer2019continuous}. 
\end{remark}

\subsection{Error Dynamics}

The goal is to design a controller for system \eqref{eq:nonl_system} to track a reference trajectory planned using its approximation \eqref{lowFidelity}. In order to do so, we proceed by deriving the evolution of the error between \eqref{eq:nonl_system} and \eqref{lowFidelity}.  
Since $\hat{n}_x \leq n_x$ in general, we define a {$\mathcal{C}^1$} map $\pi : \R^{\hat{n}_x} \to \R^{n_x}$, called the \textit{comparison} map, and we define the tracking error as:
\begin{align}\label{eq:trac_err}
e(t) = x(t) - \pi(\hat{x}(t)).
\end{align}
If the planning model is simply a linearization, we may select $\pi$ to be the identity map, but our primary interest is in the case where $\hat{x}$ is of lower dimension and $\pi$ lifts it to the dimension of $x$.
We will first describe the method with this simple error definition to build intuition before generalizing the error definition in Section~\ref{sec:input_jump}, where $\pi$ is allowed to also depend on $\hat{u}$. 

Differentiating \eqref{eq:trac_err} with respect to time (dropping time arguments to improve readability), and eliminating the variable $x$, we obtain:
\begin{align}
\dot{e} &= \dot{x} - \frac{\partial \pi}{\partial \hat{x}} \cdot \dot{\hat{x}} \nonumber \\
&= \left. f(x,w) + g(x,w) \cdot u - \frac{\partial \pi}{\partial \hat{x}} \cdot \hat{f}(\hat{x}, \hat{u}) \right|_{x = e + \pi(\hat{x})}, \nonumber \\
&= f_e(e, \hat{x}, \hat{u}, w) + g_e(e, \hat{x}, w) \cdot u, \label{errorDyn}
\end{align}
where we have defined:
\begin{align}
f_e(e, \hat{x}, \hat{u},w) & = f(\pi(\hat{x}) + e,w) - \frac{\partial \pi}{\partial \hat{x}} \cdot \hat{f}(\hat{x}, \hat{u}), \nonumber \\
g_e(e, \hat{x},w) & = g(\pi(\hat{x}) + e,w). \label{errorAffine}
\end{align}

\begin{assumption}\label{ass:initial}
The initial condition of error-state, $e(0)$, starts within a set $\mathcal{E}_0  \subset \R^{n_x}$, i.e., $e(0) \in \mathcal{E}_0$. 
\end{assumption}
For this paper, we consider a parameterization of the tracking controller given by: 
\begin{equation} \label{lowLevel}
u(t) = \kappa(e(t), \hat{x}(t), \hat{u}(t)),~\kappa \in \mathcal{K}_\mathcal{U}
\end{equation}
where the set $\mathcal{K}_{\mathcal{U}}:= \{\kappa: \R^{n_x} \times \R^{\hat{n}_x} \times \R^{\hat{n}_u} \rightarrow \mathcal{U}\}$ defines a set of admissible error-state feedback control laws. The tracking controller \eqref{lowLevel} is to be designed such that $e(t) \in \mathcal{E}$ for a bounded set $\mathcal{E}$, for all $t \geq 0$, evolving with the dynamics \eqref{errorDyn} and \eqref{lowLevel}. This set $\mathcal{E}$ is called the Robust Infinite-Time Forward Reachable Set of $\mathcal{E}_0$, and is formally defined next. 
\begin{definition}[Robust Infinite-Time Forward Reachable Set]\label{def:TEB}
Consider \eqref{errorDyn} in closed-loop with \eqref{lowLevel} for all $t \ge 0$ as: 
\begin{align} \label{eq:err_closedloop}
    &\dot{e} = f_e(e,\hat{x},\hat{u}, w) + g_e(e,\hat{x}, w)\cdot \kappa(e, \hat{x},\hat{u}), 
\end{align}
with $\hat{x}$, $\hat{u}$, and $w$ constrained by $\hat{\mathcal{X}}$, $\hat{\mathcal{U}}$, and $\mathcal{W}$. Then a \emph{robust infinite-time forward reachable set} $\mathcal{E}$ of $\mathcal{E}_0$ is defined as: 
\begin{align*}
    &\mathcal{E} := \{e(t) \in \R^{n_x}: \exists \ e(0) \in \mathcal{E}_0, \ \hat{x}: \R_+ \rightarrow \hat{\mathcal{X}},~ \hat{u}: \R_+ \rightarrow \hat{\mathcal{U}}, \\
    &~~~~~~~~~~w: \R_+ \rightarrow \mathcal{W},\ t \ge 0, ~ \text{s.t.} ~ e(t) ~ \text{is a solution to}~ \eqref{eq:err_closedloop}\}.
\end{align*}
\end{definition}
As computing the robust infinite-time forward reachable set $\mathcal{E}$
is intractable in general,
we find a tracking control law $\kappa$ and compute an outer-bound $\mathcal{O} \supseteq \mathcal{E}$. We refer to $\mathcal{O}$ as a {``tracking error bound'' (TEB)}, and $\kappa$ as the corresponding ``tracking controller". {A stylized depiction of $\mathcal{E}_0$, $\mathcal{E}$, and $\mathcal{O}$ is shown in Figure~\ref{fig: TEB}.}
\begin{figure}[tbh]
    \centering
    \includegraphics[width=0.25 \textwidth]{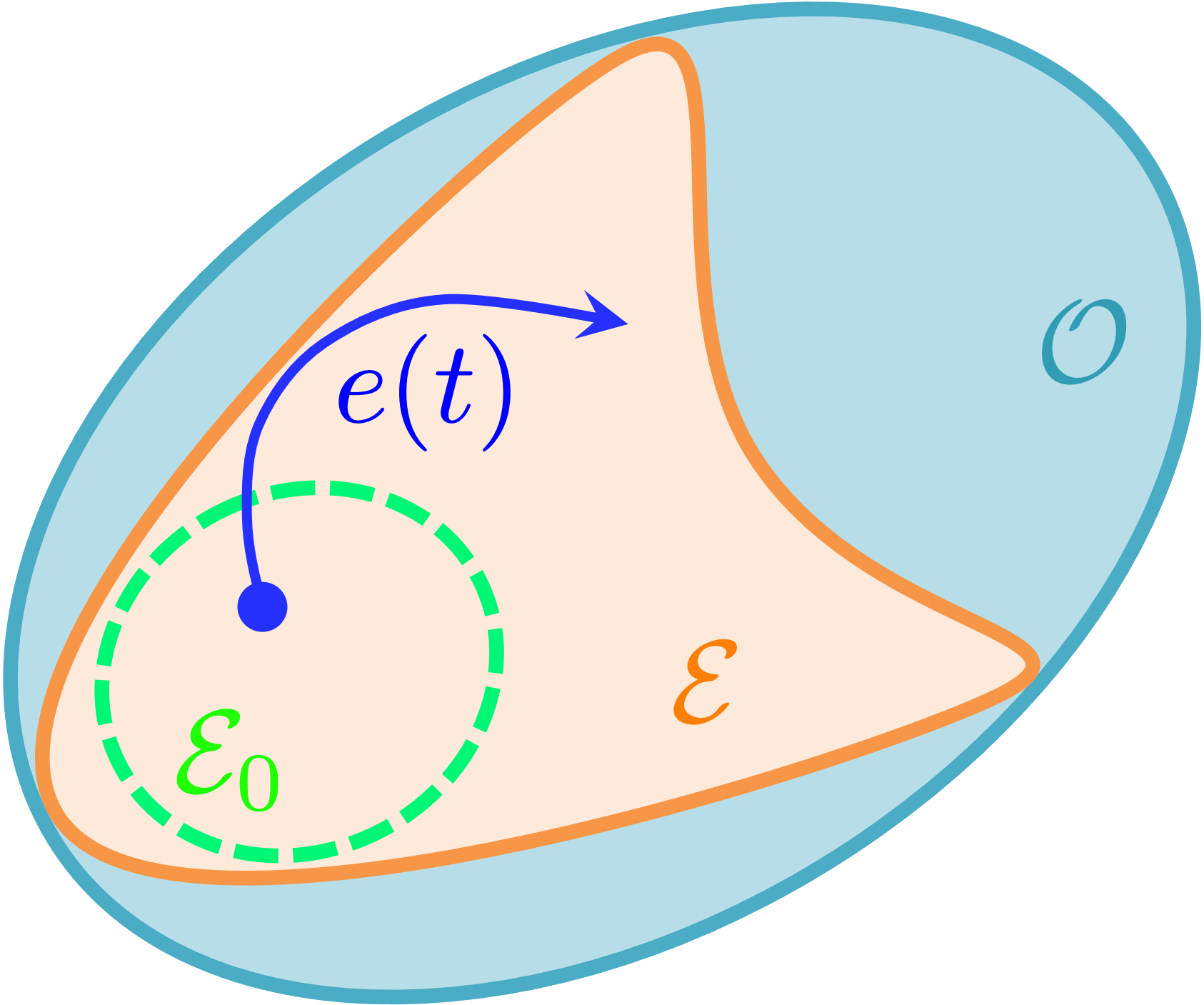}
    \caption{Illustration of Definition~\ref{def:TEB}, with initial error set $\mathcal{E}_0$, error trajectory $e(t)$, and robust infinite-time forward reachable set $\mathcal{E}$. The TEB $\mathcal{O}\supseteq\mathcal{E}$ is also shown.}
    \label{fig: TEB}
\end{figure}

{As we will see in the next subsection, we aim to minimize the volume of the set $\mathcal{O}$ when designing the tracking controller $\kappa$; however, we do not emphasize asymptotic behavior of the error $e(t)$  since we do not need perfect tracking of the planning model. Indeed we allowed the dynamics (\ref{eq:err_closedloop}) to depend on $\hat{x}$ besides $\hat{u}$ and $w$, and we did not require the right-hand side to vanish when $e=0$. The benefit of this relaxed approach, as alluded to in the Introduction, is to remove restrictive geometric constraints from the selection of the map $\pi$ and controller $\kappa$ that would render the set $e=x-\pi(\hat{x})=0$ invariant and attractive.} 

\subsection{Computing the TEB and Tracking Controller}
The TEB $\mathcal{O}$ and the tracking controller $\kappa$ can be obtained with the help of the following theorem.  
\begin{theorem} \label{thm1}
Given the error dynamics \eqref{errorDyn} with mapping $f_e : \R^{n_x} \times \R^{\hat{n}_x} \times \R^{\hat{n}_u} \times \R^{n_w} \to \R^{n_x}$, $g_e : \R^{n_x} \times \R^{\hat{n}_x} \times \R^{n_w} \to \R^{n_x}$, $\gamma \in \R$, $\hat{\mathcal{X}} \subseteq \R^{\hat{n}_x}$, $\hat{\mathcal{U}} \subseteq \R^{\hat{n}_u}$, and $\mathcal{W} \subseteq \R^{n_w}$, if there exists a $\mathcal{C}^1$ function $V: \R^{n_x} \to \R$ and $\kappa : \R^{n_x} \times \R^{\hat{n}_x} \times \R^{\hat{n}_u} \to \R^{n_u}$ such that
\begin{subequations}\label{eq:them1cond}
\begin{align}
& \mathcal{E}_0 \subseteq \{e \in \R^{n_x} : V(e) \leq \gamma\}, \label{eq:thm1cond1}\\
&\frac{\partial V(e)}{\partial e} \cdot \left( f_e(e, \hat{x}, \hat{u}, w) + g_e(e, \hat{x}, w) \cdot \kappa(e, \hat{x}, \hat{u}) \right) <  0, \nonumber \\
& \forall e, \hat{x}, \hat{u}, w, \ \text{s.t.} \ V(e) = \gamma, \ \hat{x} \in \hat{\mathcal{X}},  \ \hat{u} \in \hat{\mathcal{U}},\ w \in \mathcal{W}, \label{eq:thm1cond2} \\
& \{e \in \R^{n_x} : V(e) \leq \gamma\} \subseteq \{e \in \R^{n_x}: \kappa(e,\hat{x},\hat{u}) \in \mathcal{U}\}, \nonumber \\
& \forall (\hat{x},\hat{u}) \in \hat{\mathcal{X}} \times \hat{\mathcal{U}}
\end{align}
\end{subequations}
hold, then the sublevel set $\Omega(V, \gamma) := \{e \in \mathbb{R}^{n_x}:V(e)\leq \gamma \}$ is a TEB, denoted by $\mathcal{O}$, achieved by the tracking control law $\kappa$.
\end{theorem}

{
\begin{proof}

The theorem is proved by contradiction. Assume there exist a time $t_2 >0$, an initial condition $e_0 \in \mathcal{E}_0$, and a trajectory $e(\cdot)$ such that $e(0) = e_0$, and $V(e(t_2)) > \gamma$. Since $V(e(0)) \leq \gamma$ from \eqref{eq:thm1cond1}, by continuity of $V$ there exists $t_1$ such that $0 \leq t_1 < t_2$, $V(e(t_1)) = \gamma$, and $\frac{d}{dt}V(e(t))|_{t=t_1} \geq 0$. (If all crossings of $V(e(t)) = \gamma $ satisfied $\frac{d}{dt}V(e(t)) < 0$, then $V$ would not be continuous.) This contradicts~(\ref{eq:thm1cond2}). 
\end{proof}
}

Finding generic functions $V$ and $\kappa$ that satisfy constraints~\eqref{eq:them1cond} is a difficult problem. 
Below we show how SOS programming can be used to search for these functions by restricting to polynomial candidates $V \in \R[e]$ and $\kappa \in \R^{n_u}[(e,\hat{x},\hat{u})]$. Besides this restriction, we make the following assumption:
\begin{assumption}\label{ass:poly}
The mappings $f_e \in \R^{n_x}[(e,\hat{x},\hat{u},w)]$ and $g_e \in \R^{n_x \times n_u}[(e,\hat{x},w)]$ in error dynamics~\eqref{errorDyn} are polynomials. Sets $\mathcal{E}_0$, $\hat{\mathcal{X}}$, $\hat{\mathcal{U}}$, and $\mathcal{W}$ are semi-algebraic sets, i.e., there exists $p_0 \in \R[e]$ such that $\mathcal{E}_0 = \{e \in \R^{n_x}~:~p_0(e) \leq 0\}$; with similar definitions for $\hat{\mathcal{X}}$, $\hat{\mathcal{U}}$, and $\mathcal{W}$ with polynomials $p_{\hat x} \in \R[\hat{x}]$, $p_{\hat u} \in \R[\hat{u}]$, and $p_w \in \R[w]$. The control constraint set $\mathcal{U}$ is a hypercube $\mathcal{U} = \{u \in \R^{n_u}: \underline{u} \leq u \leq \overline{u}\}$, where $\underline{u}, \overline{u} \in \R^{n_u}$.
\end{assumption}


By applying the generalized S-procedure~\cite{Parrilo:00} to the set containment constraints~\eqref{eq:them1cond}, and using the volume of $\Omega(V, \gamma)$ as the cost function to minimize, we obtain the following SOS optimization problem for finding $V$ and $\kappa$:
\begingroup
\allowdisplaybreaks
\begin{subequations} \label{eq:sosopt1}
\begin{align}
\displaystyle \min_{V, \kappa, s,l} ~&
        \text{volume}(\Omega(V, \gamma)) \vspace{2mm} \nonumber \\ 
        \ \  \text{s.t.} ~
    &s_0 \in \Sigma[e], s_{1\rightarrow3} \in \Sigma[(e,\hat{x},\hat{u},w)], l \in \R[(e,\hat{x},\hat{u},w)] \nonumber \\
    & s_{4 \rightarrow 9, i} \in \Sigma[(e,\hat{x},\hat{u})], i \in \{1, ..., n_u\}  \label{eq:sos_cond0}\\
    &-(V(e) - \gamma) + s_0\cdot p_0 \in \Sigma[e], \label{eq:sos_cond1}\\
    &-\frac{\partial V}{\partial e}\cdot(f_e+g_e \cdot \kappa) - \epsilon e^\top e +l \cdot (V - \gamma) + s_1 \cdot p_{\hat x} \nonumber \\
    & ~~~~~~ +s_2 \cdot p_{\hat u} + s_3 \cdot p_w \in \Sigma[(e,\hat{x},\hat{u},w)], \label{eq:sos_cond2} \\
    &\overline{u}_i - \kappa_i + s_{4,i}\cdot (V - \gamma)+ s_{5,i} \cdot p_{\hat{x}} \nonumber \\
    & ~~~~~~  + s_{6,i} \cdot p_{\hat{u}} \in \Sigma[(e,\hat{x},\hat{u})],  i \in \{1, ..., n_u\}, \label{eq:sos_cond3} \\
    &\kappa_i - \underline{u}_i + s_{7,i}\cdot (V - \gamma)+ s_{8,i} \cdot p_{\hat{x}} \nonumber \\
    & ~~~~~~  + s_{9,i} \cdot p_{\hat{u}} \in \Sigma[(e,\hat{x},\hat{u})], i \in \{1, ..., n_u\}. \label{eq:sos_cond4}
\end{align}
\end{subequations}
\endgroup
In the formulation above, SOS polynomials $s_{1\rightarrow3}$ and $s_{4\rightarrow9,i}$ are multipliers used in the generalized S-procedure, and $\epsilon > 0$ is on the order of $10^{-6}$. The optimization~\eqref{eq:sosopt1} is non-convex as there are two groups of decision variables $V$ and $(\kappa, l, s_{4,i}, s_{7,i})$ bilinear in each other. To tackle this problem, similarly to \cite[Algorithm 1]{YinBackward19}, we decompose it into two tractable subproblems to iteratively search between the two groups of decision variables, as shown in Algorithm~\ref{alg:alg1} in the Appendix.

\subsection{Safety Check}
After synthesizing $V$ and $\kappa$, we check the following safety condition with $\mathcal{O} = \Omega(V,\gamma)$:
\begin{align}\label{eq:safetyCond1}
    \pi(\hat{\mathcal{X}}) \oplus \mathcal{O} \subseteq \mathcal{X}.
\end{align}
If (\ref{eq:safetyCond1}) is satisfied, then the tracker state $x$ is guaranteed to satisfy state constraints $\mathcal{X}$ and the design is considered successful. If (\ref{eq:safetyCond1}) is \textit{not} satisfied, we shrink the planner sets $\hat{\mathcal{X}}$ and $\mathcal{U}$ and repeat the process as indicated in Figure~\ref{fig: diagram}. 

\section{Generalized Tracking Error Definition} \label{sec:input_jump}
So far, we have used a map $\pi$ that only depends on the planning state $\hat{x}$ in \eqref{eq:trac_err}. However, as illustrated in the example below, this map may fail to provide reference signals for all the tracker states. 
Therefore, in Section~\ref{subsec:modError},
we move to a more general error definition that also depends on the planner input $\hat{u}$. 

\begin{example} \label{ex:integrators}
As a simple illustration of why using a kinematic model for planning necessitates a more general error definition, consider the tracking model
\begin{align}
    x = \begin{bmatrix}s\\v\end{bmatrix}, \hspace{1em} \dot{x} =\begin{bmatrix}v\\u\end{bmatrix}, 
\end{align}
where $s$ is the position, $v$ is the velocity, and $u$ is the acceleration input. Let the planning model be a single integrator, where the only state is the planner position ($\hat{x} = \hat{s}$) and the input is the planner velocity ($\dot{\hat{x}} = \hat{v} =: \hat{u}$). Then, letting $\pi(\hat{x},\hat{u}) = [\hat{x};\hat{u}]$,
the error is 
\begin{align}
    e =  x - \pi(\hat{x},\hat{u}) = \begin{bmatrix} s-\hat{s}\\ v - \hat{v}\end{bmatrix},
\end{align}
which is the deviation of the planner and tracker positions and velocities. Thus, by keeping $e$ small, we keep the planner and tracker positions and velocities close to one another, which is desirable. 
On the other hand, if we had used the na{\"i}ve map $\pi(\hat{x}) = [\hat{x}; 0]$, the error would be 
$[(s-\hat{s}) \ \ v]^\top$, 
and bounding the error would mean keeping $v$ close to zero, which is overly conservative and may not align with planning objectives.
\end{example}

\subsection{Modified Error Dynamics}\label{subsec:modError}
As motivated above, we will use a more general {$\mathcal{C}^1$} map $\pi: \R^{\hat{n}_x}\times\R^{\hat{n}_u} \rightarrow \R^{n_x}$ to 
provide better reference trajectories
for the tracking model, 
as was done in~\cite{meyer2019continuous} for the first time. For further generality, in this article we redefine the error state as 
\begin{align}\label{eq:error_xhat_uhat}
    e = \phi(x,\hat{x},\hat{u})(x-\pi(\hat{x},\hat{u})),
\end{align}
where we add the $\mathcal{C}^1$ map $\phi:  \R^{{n}_x}\times\R^{\hat{n}_x}\times\R^{\hat{n}_u} \rightarrow \R^{n_x\times n_x}$ which provides additional flexibility, as will be demonstrated in Section~\ref{examples}.

Assume that for each $e$, $\hat{x}$, $\hat{u}$, there exists a unique $x$ satisfying (\ref{eq:error_xhat_uhat}), and denote this inverse as 
\begin{align}
    x = \nu(e,\hat{x},\hat{u}).
\end{align}
The error dynamics resulting from (\ref{eq:error_xhat_uhat}) are 
\begin{align}
    \label{eq:error_dyn2}
    \dot{e} = f_e(e, \hat{x}, \hat{u},w) + g_e(e, \hat{x}, \hat{u},w) u - h_e(e,\hat{x},\hat{u})\dot{\hat{u}},
\end{align}
where 
\begin{align}
    &{f}_e(e,\hat{x},\hat{u},w):=  \left\{\frac{\partial \phi}{\partial x}f(x,w) + \frac{\partial \phi}{\partial \hat{x}}\hat{f}(\hat{x},\hat{u})\right\}(x-\pi(\hat{x},\hat{u})) \nonumber\\
    &~~~~\ + \phi(x,\hat{x},\hat{u}) \left\{f(x,w) - \frac{\partial\pi}{\partial\hat{x}}\hat{f}(\hat{x},\hat{u})\right\}\bigg\vert_{x = \nu(e,\hat{x},\hat{u})}, \label{eq:fe}\\
    &{g}_e(e,\hat{x},\hat{u},w):= \bigg\{\frac{\partial \phi}{\partial x}(x-\pi(\hat{x},\hat{u})) \label{eq:ge}\\
    &~~~~~~~~~~~~~~~~~~~~~~~~~~~~~~\ + \phi(x,\hat{x},\hat{u})\bigg\} g(x,w) \bigg\vert_{x = \nu(e,\hat{x},\hat{u})}, \nonumber
\end{align}
and $h_e$ can be computed but is not written explicitly since it multiplies $\dot{\hat{u}}$, which is zero within sampling periods.

\subsection{Analysis Within Sampling Periods}
Note that the planner input is applied in a zero order hold fashion 
within each sampling period
as described in \eqref{eq:zero_order_uhat}. As the tracking error dynamics in \eqref{eq:error_dyn2} has a term containing $\dot{\hat{u}}$ (unlike \eqref{errorDyn}), these dynamics change discontinuously at each sampling instant $\tau_k$. Therefore, instead of considering a tracking controller for all times, we consider only a time interval between any two sampling instants. {For additional flexibility we consider a time-varying tracking controller.} 
Since the signal $\hat{u}$ is piece-wise constant, we thus have 
\begin{align}
\dot{\hat{u}}(t) = 0, \ \forall t \in [\tau_k, \tau_{k+1}). \label{eq:uhatdot}
\end{align}
Therefore, the error dynamics~\eqref{eq:error_dyn2} during the time interval $[\tau_k, \tau_{k+1})$ are:
\begin{equation}
    \label{eq:error_dyn3}
    \dot{e} = f_e(e, \hat{x}, \hat{u},w) + g_e(e, \hat{x}, \hat{u},w) u.
\end{equation}
Given the bounded set of initial conditions $\mathcal{E}_0$, we want to enforce the boundedness of the error state during $[0, T_s)$ by introducing a tracking controller
\begin{align}
\label{eq: TVcontroller}
u(t) = \kappa(t, e(t), \hat{x}(t), \hat{u}(t)),
\end{align}
which is now defined by a \emph{time-varying}, error-state feedback control law $\kappa: \R \times \R^{n_x} \times \R^{\hat{n}_x} \times \R^{\hat{n}_u} \rightarrow \R^{n_u}$.
Below, we provide 
the design requirements on $\kappa$ to obtain such an error bound. 
\begin{proposition}
\label{prop time varying}
Given the error dynamics~\eqref{eq:error_dyn3} with mappings $f_e: \mathbb{R}^{n_x} \times \mathbb{R}^{\hat{n}_x} \times \mathbb{R}^{\hat{n}_u} \times \R^{n_w} \rightarrow \mathbb{R}^{n_x}$, $g_e: \mathbb{R}^{n_x} \times \mathbb{R}^{\hat{n}_x} \times \R^{\hat{n}_u}\times \R^{n_w} \rightarrow \mathbb{R}^{n_x}$, 
and $\gamma \in \R$, $T_s > 0$, $\hat{\mathcal{X}} \subseteq \mathbb{R}^{\hat{n}_x}$, $\hat{\mathcal{U}} \subseteq \mathbb{R}^{\hat{n}_u}$, $\mathcal{W} \subseteq \mathbb{R}^{n_w}$, if there exists a $\mathcal{C}^1$ function $V: \R \times \mathbb{R}^{n_x} \rightarrow \mathbb{R}$, and $\kappa: \R \times \mathbb{R}^{n_x} \times \mathbb{R}^{\hat{n}_x} \times \mathbb{R}^{\hat{n}_u} \rightarrow \mathbb{R}^{n_u}$, such that 
\begin{subequations}
\begin{align}
    &\mathcal{E}_0 \subseteq \{e \in \R^n: V(0,e)\leq \gamma\}, \label{eq:V_cond0}\\
    & \frac{\partial V(t,e)}{\partial e}\cdot (f_e(e, \hat{x}, \hat{u},w)+ g_e(e, \hat{x}, \hat{u},w)\cdot \kappa(t, e,\hat{x},\hat{u}))\nonumber \\
    &\ ~~~~ + \frac{\partial V(t,e)}{\partial t} < 0, \ \forall t, e, \hat{x}, \hat{u}, w, \ \text{s.t.} \ t\in[0,T_s), \nonumber \\
	&\ \ ~~~~ V(t, e) = \gamma, \ \hat{x} \in \hat{\mathcal{X}}, \ \hat{u} \in \hat{\mathcal{U}},\ w \in \mathcal{W}, \label{eq:V_cond1} \\
	& \{e \in \R^{n_x} : V(t, e) \leq \gamma\} \subseteq \{e \in \R^{n_x}: \kappa(t, e,\hat{x},\hat{u}) \in \mathcal{U}\}, \nonumber \\
	&\ \ ~~~~ \forall (t, \hat{x},\hat{u}) \in [0, T_s)\times  \hat{\mathcal{X}} \times \hat{\mathcal{U}} \label{eq:V_cond2}
	\end{align}
\end{subequations}
then for all $e(0) \in \mathcal{E}_0$, we have $e(t) \in \{e \in \R^{n_x}: V(t,e)\leq \gamma\}$, for all $t \in [0,T_s)$.
\end{proposition}
\begin{proof}
The proof follows  the proof of Theorem~\ref{thm1}.
\end{proof}

Define the funnel $\Omega(V,t,\gamma):= \{e \in \R^{n_x}:V(t,e)\leq \gamma\}$. Then Proposition~\ref{prop time varying} can be restated as 
\begin{align}
    e(0) \in \mathcal{E}_0 \subseteq \Omega(V,0,\gamma) \Rightarrow e(t) \in \Omega(V,t,\gamma), \ \forall t \in [0, T_s).\nonumber
\end{align}

\begin{remark} \label{remark:shift_V}
Although Proposition~\ref{prop time varying} is stated for the first sampling period $[0,T_s)$, it can be used for any other sampling period $[\tau_k,\tau_{k+1})$ with $\tau_k=k \cdot T_s$. Let $e(\tau_{k}) \in \Omega(V,0,\gamma)$. Then we have $e(\tau_{k}+t) \in \Omega(V,t,\gamma)$, for all $t \in [0, T_s)$, under the control signal $u(\tau_k + t)=\kappa(t,e(\tau_k+t),\hat{x}(\tau_k+t),\hat{u}(\tau_k+t)))$.
\end{remark}

\subsection{Analysis Across Sampling Periods}
Next, we focus on the effect of the input jump $\Delta\hat{u}$ at 
each sampling instant $\tau_k$
as in \eqref{eq:uhat_jump}.
From (\ref{eq:error_xhat_uhat}), $\Delta\hat{u}$ induces a jump on the error. 
Let $\tau_{k}^-$ and $\tau_{k}^+$ denote sampling instant $\tau_{k}$ before and after the discrete jump, respectively, {and for simplicity use the notation $e^+_k := e(\tau_k^+)$ and so on.} Then we have
{
\begin{align}
    e^+_k &= \phi(x^+_k,\hat{x}^+_k,\hat{u}^+_k)(x^+_k - \pi(\hat{x}^+_k,\hat{u}^+_k)) \nonumber\\
    &= \phi(x^-_k,\hat{x}^-_k,\hat{u}^-_k + \Delta\hat{u}^+_k)(x^-_k - \pi(\hat{x}^-_k,\hat{u}^-_k + \Delta\hat{u}^+_k)) \nonumber\\
    &= \phi(\nu(e^-_k,\hat{x}^-_k,\hat{u}^-_k + \Delta\hat{u}^+_k),\hat{x}^-_k,\hat{u}^-_k + \Delta\hat{u}^+_k) \nonumber\\
    &~~~~~~~~~~\ \cdot (\nu(e^-_k,\hat{x}^-_k,\hat{u}^-_k + \Delta\hat{u}^+_k) - \pi(\hat{x}^-_k,\hat{u}^-_k + \Delta\hat{u}^+_k)) \nonumber\\
    &=: h(e^-_k,\hat{x}^-_k,\hat{u}^-_k,\Delta\hat{u}^+_k). \label{eq:jump}
\end{align}
We refer to $h$ as the \textit{jump function}, as it reflects how the error may jump from the end of one sampling period to the beginning of the next, due to the jump in the input.

We introduce the additional condition below to characterize the error jump induced by the control jump $\Delta\hat{u}$ in terms of the funnel $\Omega(V,t,\gamma)$.


\begin{proposition} \label{prop_jump_general}
Given $\gamma \in \R$, $T_s\in\R$,  
$\hat{\mathcal{X}}\subseteq\mathbb{R}^{\hat{n}_x}$, $\hat{\mathcal{U}}\subseteq\mathbb{R}^{\hat{n}_u}$, $\Delta \hat{\mathcal{U}} \subseteq \R^{\hat{n}_u}$, 
$h:\mathbb{R}^{\hat{n}_x}\times\mathbb{R}^{\hat{n}_x}\times\mathbb{R}^{\hat{n}_u}\times\mathbb{R}^{\hat{n}_u}\rightarrow\mathbb{R}^{n_x}$, 
if there exists a function $V: \R \times \R^{n_x} \rightarrow \R$ satisfying
\begin{align}\label{eq:V_cond3} 
    &V(0, h(e,\hat{x},\hat{u},\Delta\hat{u})) \leq \gamma, \\
    &~~~~~~~~\ \forall \hat{x}\in\hat{\mathcal{X}},\hat{u}\in \hat{\mathcal{U}}, \Delta\hat{u} \in \Delta \hat{\mathcal{U}} 
    \text{ and } \forall e \text{ s.t.} \ V(T_s,e) \leq \gamma \nonumber 
\end{align}
then for all $e^-_k \in \Omega(V,T_s,\gamma)$, $e^+_k \in \Omega(V,0,\gamma)$. 
\end{proposition}

\begin{proof}
Suppose $e^-_k \in \Omega(V,T_s,\gamma)$, i.e., $V(T_s, e^-_k)\leq \gamma$. By Eq.~\ref{eq:jump}, $e^+_k =  h(e^-_k,\hat{x}^-_k,\hat{u}^-_k,\Delta\hat{u}^+_k)$. Thus by Eq.~\ref{eq:V_cond3}, $V(0,e^+_k)\leq \gamma$, i.e., $e^+_k \in \Omega(V,0,\gamma)$.
\end{proof}

\begin{remark}
For the special case $\phi(x,\hat{x},\hat{u}) = 1$ and $\pi(\hat{x},\hat{u}) = \theta(\hat{x})+Q\hat{u}$ {for some $Q\in\mathbb{R}^{n_x\times\hat{n}_u}$}, the $\hat{x}$ and $\hat{u}$ terms cancel, and so (\ref{eq:V_cond3}) simplifies to
\begin{align}
    V(0,e-Q\Delta\hat{u}) \leq \gamma, \  \forall \Delta\hat{u}\in\Delta\hat{\mathcal{U}}, \ \forall e \text{ s.t. } V(T_s,e)\leq \gamma. \nonumber
\end{align}
\end{remark}
}

\subsection{Combining Within- and Across-Sample Analysis}
We next combine the conditions for within- and across-sample error boundedness
from Propositions~\ref{prop time varying} and~\ref{prop_jump_general}, respectively, to obtain the main result on the boundedness of the error at all time, formulated below and illustrated in Figure~\ref{fig: funnel}.

\begin{theorem}
\label{thm2}
If there exist $V$ and $\kappa$ satisfying \eqref{eq:V_cond0}--\eqref{eq:V_cond2}, and \eqref{eq:V_cond3}, define $\mathcal{O} \subset \R^{n_x}$ such that $$\cup_{t\in[0,T_s)}\Omega(V,t,\gamma) \subseteq \mathcal{O}.$$
Then for all $\hat{x}(t) \in \hat{\mathcal{X}}$, $\hat{u}(t) \in \hat{\mathcal{U}}$, $\Delta\hat{u}(t) \in \Delta \hat{\mathcal{U}}$, and $w(t) \in \mathcal{W}$, the error system (\ref{eq:error_dyn2}) under control law $u(t)=\kappa(\tilde{t},e(t),\hat{x}(t),\hat{u}(t)))$ with $\tilde t=(t\mod T_s)\in[0,T_s)$ satisfies:
$$e(0) \in \mathcal{E}_0~\Rightarrow ~e(t)\in \mathcal{O}, ~\forall t \ge 0,$$
that is to say, $\mathcal{O}$ is a TEB achieved by the tracking control law $\kappa$.
\end{theorem}
\begin{proof}
From Remark~\ref{remark:shift_V} and for all $\tau_k = k\cdot T_s$, we have if $e(\tau_k) \in \Omega(V,0,\gamma)$, then $e(\tau_k + \tilde{t}) \in \Omega(V,\tilde{t},\gamma)$ and $e(\tau_{k+1}^-) \in \Omega(V,T_s,\gamma)$. Then it follows from Proposition \ref{prop_jump_general} that $e(\tau_{k+1}^+) \in \Omega(V,0,\gamma)$. As a result, for all $e(0) \in \mathcal{E}_0 \subseteq \Omega(V,0,\gamma)$, we have $e(k \cdot T_s +\tilde{t}) \in \Omega(V,\tilde{t},\gamma) \subseteq \mathcal{O}$, for all $k \ge 0$, and $\tilde{t} \in [0, T_s)$.
\end{proof}

\begin{figure}[tbh]
    \centering
    \includegraphics[width=0.45 \textwidth]{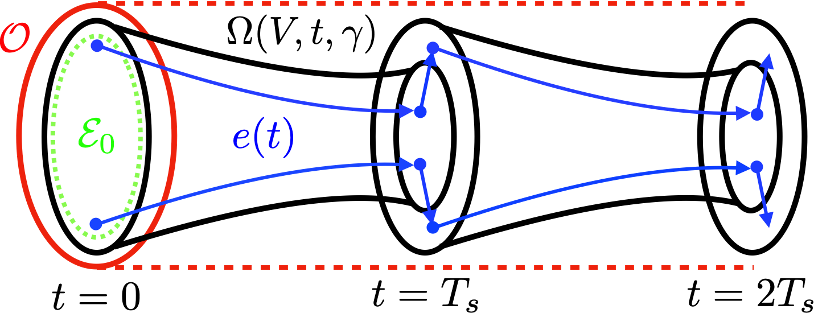}
    \caption{Illustration of Theorem~\ref{thm2}, with initial error set $\mathcal{E}_0$, funnels $\Omega_{t,\gamma}^V$ on each sampling period, bounded error jumps at sampling times, and TEB $\mathcal{O}$.}
    \label{fig: funnel}
\end{figure}

{
\begin{example} \label{ex:simpleStorageFn}
We now construct a storage function for the planner/tracker dynamics in Example~\ref{ex:integrators}.
Within sampling periods, $\dot{\hat{u}} = 0$, so the open loop error dynamics are 
\begin{align}
    \dot{e} =\begin{bmatrix} \dot{s}-\dot{\hat{s}}\\ \dot{v} - \dot{\hat{v}}\end{bmatrix} =  \begin{bmatrix} v-\hat{v}\\ u - 0\end{bmatrix} = \begin{bmatrix}e_2 \\ u \end{bmatrix}.
\end{align} 
Selecting a state-feedback controller $u(e) = -k_1 e_1 - k_2 e_2$, the closed loop error dynamics are
\begin{align}
    \dot{e} = \begin{bmatrix} 0 & 1\\ -k_1 & -k_2\end{bmatrix} e =: Ae. 
\end{align}
Because this is a LTI system, constructing a storage function is straightforward. If there exists $P = P^\top > 0$ such that {$PA + A^\top P < -\alpha P$} for some $\alpha>0$, then it is simple to show that $V(t,e) = \exp(\alpha t) \cdot e^\top P e$ satisfies 
$\dot{V}(t,e) < 0$ for all $e\neq 0$. 
Hence, for any $\gamma > 0$, $\cup_{t\in[0,T_s)}\Omega(V,t,\gamma)$ forms a valid TEB by Theorem~\ref{thm2}.
We can examine the form of the level sets to see how they shrink with time
\begin{align}
    \Omega(V,t,\gamma) 
    &= \left\{e\in\mathbb{R}^{n_x} : e^\top P e \leq \frac{\gamma}{\exp{\alpha t}}\right\}
\end{align}
As $t$ increases, $e$ is forced to lie in smaller and smaller ellipsoids. Then the jump condition is 
\begin{align}
    &\left(e-\begin{bmatrix}0\\ \Delta\hat{u}\end{bmatrix}\right)^\top P\left(e-\begin{bmatrix}0\\ \Delta\hat{u}\end{bmatrix}\right) \leq \gamma \\
    \text{ for all } & \Delta\hat{u} \in\Delta\hat{\mathcal{U}} \text{ and } e \text{ s.t. } e^\top P e \leq \frac{\gamma}{\exp(\alpha T_s)}, \nonumber
\end{align}
meaning that if the error lies in the smallest ellipsoid at the end of the sampling period, then for all values of $\Delta\hat{u}$ the perturbed error will lie in the largest ellipsoid at the start of the next sampling period, as illustrated in Figure~\ref{fig: funnel}. 
\end{example}
}

\subsection{SOS Optimization}
Again, to use SOS optimization to search for $V$ and $\kappa$, we restrict them to polynomials: $V \in \R[(t,e)]$, and $\kappa \in \R[(t,e,\hat{x},\hat{u})]$. 
{We further assume that $f_e$ (\ref{eq:fe}), $g_e$ (\ref{eq:ge}), and the jump function $h$ (\ref{eq:jump}) are polynomials.}
In addition to Assumption~\ref{ass:poly}, we assume $\Delta\hat{\mathcal{U}} = \{\Delta \hat{u} \in \R^{\hat{n}_{u}}: p_\Delta(\Delta \hat{u}) \leq 0\}$, where $p_\Delta \in \R[\Delta \hat{u}]$. By choosing the integral of the volume of $\Omega(V,t,\gamma)$ over the time interval $[0, T_s]$ as the cost function, and applying the generalized S-procedure to \eqref{eq:V_cond0}--\eqref{eq:V_cond2}, and \eqref{eq:V_cond3}, we obtain the following optimization problem:
\begingroup
\allowdisplaybreaks
\begin{subequations} \label{eq:sosopt2}
\begin{align}
\displaystyle \min_{V, \kappa, s,l} ~&
        \int_0^{T_s} \text{volume}(\Omega(V,t,\gamma))dt \vspace{2mm} \nonumber \\ 
        \ \  \mathrm{s.t.} ~
    &s_{1\rightarrow4} \in \Sigma[(t,e,\hat{x},\hat{u},w)], s_{5\rightarrow6} \in \Sigma[(e,\Delta\hat{u})],\nonumber \\
    & l \in \R[(t,e,\hat{x},\hat{u},w)], s_0 \in \Sigma[e], \nonumber \\
    & s_{7\rightarrow14,i} \in \Sigma[(t,e,\hat{x},\hat{u})], i \in \{1,...,n_u\}, \label{eq:opt2_cond0}\\
    &\gamma-V(0,e) + s_0\cdot p_0 \in \Sigma[e], \label{eq:opt2_cond1}\\
    &-\left(\frac{\partial V}{\partial t}+\frac{\partial V}{\partial e}\cdot(f_e+g_e \kappa)\right) - \epsilon e^\top e +l \cdot (V - \gamma)  \nonumber \\
    &~~~ + s_1 \cdot p_{\hat x}+s_2 \cdot p_{\hat u} +s_3 \cdot p_w  - s_4 \cdot t (T_s - t) \nonumber \\
    &~~~ \in \Sigma[(t,e,\hat{x},\hat{u},w)], \label{eq:opt2_cond2} \\
    &-(V(0,e - P\cdot [0;\Delta \hat{u}])-\gamma) + s_5\cdot(V(T_s,e)-\gamma) \nonumber \\
    &~~~ +s_6 \cdot p_{\Delta} \in \Sigma[(e,\Delta\hat{u})], \label{eq:opt2_cond3}\\
    &\overline{u}_i - \kappa_i + s_{7,i}\cdot (V - \gamma)- s_{8,i} \cdot t (T_s - t) + s_{9,i} \cdot p_{\hat{x}} \nonumber \\
    & ~~~  + s_{10,i} \cdot p_{\hat{u}} \in \Sigma[(t, e,\hat{x},\hat{u})], i \in \{1, ..., n_u\}, \label{eq:opt2_cond4} \\
    &\kappa_i - \underline{u}_i + s_{11,i}\cdot (V - \gamma) - s_{12,i} \cdot t (T_s - t) + s_{13,i} \cdot p_{\hat{x}} \nonumber \\
    & ~~~  + s_{14,i} \cdot p_{\hat{u}} \in \Sigma[(t,e,\hat{x},\hat{u})], i \in \{1, ..., n_u\}. \label{eq:opt2_cond5}
\end{align}
\end{subequations}
\endgroup
The optimization is bilinear in two groups of decision variables $V$ and $(\kappa, l, s_5, s_{7,i}, s_{11,i})$, and can also be solved using alternating direction method similar to Algorithm~\ref{alg:alg1} in the Appendix.

After the funnel $\Omega(V,t,\gamma)$ is found, the next step is to  compute a TEB $\mathcal{O}$ by solving a convex optimization:
\begin{equation}\label{eq:TEB_compute}
\begin{aligned}
    \min \ \ &\text{volume}(\mathcal{O}) \\
    \text{s.t.} \ \ & \Omega(V,t,\gamma) \subseteq \mathcal{O}, \ \forall t \in [0, T_s].  
\end{aligned}
\end{equation}
The set $\mathcal{O}$ is restricted  to a semi-algebraic set in order to convert the set containment constraint into an SOS constraint. Depending on the parameterization of $\mathcal{O}$, different cost functions can be chosen. For example, if $\mathcal{O}$ is an ellipsoid, $\mathcal{O} = \{e \in \R^{n_x}: e^\top P_{\mathcal{O}} e \leq 1\}$, where $P_{\mathcal{O}} \in \mathbb{S}^{n_x}_{++}$ is a decision variable, then $-\log\det(P_{\mathcal{O}})$ can be used as a cost function. If $\mathcal{O}$ is a polytope, $\mathcal{O} = \{e \in \R^{n_x}: A_{\mathcal{O}} e  \leq b_{\mathcal{O}}\}$, where $A_{\mathcal{O}} \in \R^{n_{\mathcal{O}}\times n_x}$ is fixed, and $b_{\mathcal{O}} \in \R^{n_{\mathcal{O}}}$ is a decision variable, then $\sum_{i=1}^{n_{\mathcal{O}}} b_{\mathcal{O}, i}$ can be used as a cost function, where $b_{\mathcal{O}, i}$ is the $i$-th element of $b_{\mathcal{O}}$.

{
Once a TEB $\mathcal{O}$ is computed from the SOS optimization \eqref{eq:sosopt2}-\eqref{eq:TEB_compute},
we can check the following safety condition, which is a generalized version of (\ref{eq:safetyCond1}):
\begin{align}\label{eq:safetyCond2}
    \nu(\mathcal{O}, \hat{\mathcal{X}}, \hat{\mathcal{U}}) \subseteq \mathcal{X}.
\end{align}
If (\ref{eq:safetyCond2}) is satisfied, then the tracker state $x$ is guaranteed to satisfy state constraints $\mathcal{X}$ and the design is considered successful. If (\ref{eq:safetyCond1}) is \textit{not} satisfied, we shrink the planner sets $\hat{\mathcal{X}}$ and $\mathcal{U}$ and repeat the process.
}

\section{Vehicle Obstacle Avoidance Example} \label{examples}


{
We now apply the planner-tracker control scheme to a vehicle obstacle avoidance example.
For the high-fidelity tracking model, we use the dynamic bicycle model from \cite{Jkong2015}:}
\begin{equation} 
\begin{aligned}
    \dot{x}_1(t) &= x_5(t) \cos(x_3(t)) - x_6(t) \sin(x_3(t)), \\
    \dot{x}_2(t) &= x_5(t) \sin(x_3(t)) + x_6(t) \cos(x_3(t)), \\
    \dot{x}_3(t) &= x_4(t), \\
    \dot{x}_4(t) &= \frac{2}{I_z}(l_f F_{c,f}(t) - l_r F_{c,r}(t)), \\
    \dot{x}_5(t) &= x_4(t)x_6(t) + u_2(t), \\
    \dot{x}_6(t) &= -x_4(t)x_5(t) + \frac{2}{m}(F_{c,f}(t) + F_{c,r}(t))
\end{aligned}
\end{equation}
with
\begin{align}
    F_{c,f} &= C_{\alpha, f} \alpha_{f}, F_{c,r} = C_{\alpha, r}  \alpha_r \\
    \alpha_f &= \frac{x_6 + l_f x_4}{x_5} - u_1, \ \alpha_r =  \frac{x_6 - l_r x_4}{x_5}
\end{align}
where $x_1$ to $x_6$ represent $x$, $y$ positions in an inertial frame, inertial heading, yaw rate, and longitudinal and lateral speeds in the body frame. Variables $u_1$, $u_2$ represent front wheel steering angle and longitudinal acceleration, $m$ and $I_z$ denote the vehicle's mass and yaw inertia, and $l_f$ and $l_r$ represent the distance
from the center of mass of the vehicle to the front and
rear axles. $C_{\alpha, i}$ is the tire cornering stiffness, where $i \in \{f, r\}$. We use the parameter values $m = 1.67 \times 10^3$~kg, $I_z = 2.1 \times 10^3~ \text{kg}/\text{m}^2$, $l_f = 0.99$~m, $l_r = 1.7$~m, $C_{\alpha,f} = -6.1595 \times 10^4~\text{N}/\text{rad}$, and $C_{\alpha,r} = -5.2095 \times 10^4~\text{N}/\text{rad}$.

The planning model is a Dubin's vehicle model:
\begingroup
\allowdisplaybreaks
\begin{align*}
    \dot{\hat{x}}_1(t) &= \hat{u}_2(t) \cos(\hat{x}_3(t)), \\
    \dot{\hat{x}}_2(t) &= \hat{u}_2(t) \sin(\hat{x}_3(t)), \\
    \dot{\hat{x}}_3(t) &= \hat{u}_1(t),
\end{align*}
\endgroup
where $\hat{x}_1$ to $\hat{x}_3$ represent $x$, $y$ positions and heading angle, and $\hat{u}_1$ and $\hat{u}_2$ represent angular velocity and velocity. If we use the map $\pi(\hat{x}) = [\hat{x}; 0_{3\times 1}]$, where $\hat{x} = [\hat{x}_1; \hat{x}_2; \hat{x}_3]$, then $x_4$ and $x_5$ will become part of the resulting error state. As a result, the magnitude of the absolute state $x_4$ and $x_5$ will be minimized in optimization~\eqref{eq:sosopt1}, which 
is practically undesirable.
To eliminate this issue, we use a map $\pi(\hat{x},\hat{u}) = [\hat{x};\hat{u};0]$, where $\hat{u} = [\hat{u}_1; \hat{u}_2]$, which also provides reference signals for $x_4$ and $x_5$. 

The error is defined as in (\ref{eq:error_xhat_uhat}), with $\pi(\hat{x},\hat{u}) = [\hat{x}; \hat{u}; 0]$ and $\phi(\hat{x}) = \text{diag}(R^{-1}(\hat{x}_3), I_{4})$, where $R(\psi) = \left[\begin{smallmatrix} \cos(\psi) & -\sin(\psi) \\ \sin(\psi) & \cos(\psi) \end{smallmatrix} \right]$. In this example, $\phi$ allows us to replace the trigonometric functions in $\hat{x}_3$ in the error dynamics by trigonometric functions in $e_3 = (x_3 - \hat{x}_3)$, which can easily be approximated by polynomials in a certain range of $e_3$. The sampling time used in this example is $T_s = 0.1$ s. The input and input jump spaces for the planning model are $\hat{\mathcal{U}} = [-\pi/8, \pi/8] \times [2, 4]$, and $\Delta \hat{\mathcal{U}} = [-\pi/50, \pi/50] \times [-0.075, 0.075]$.

\begin{figure*}
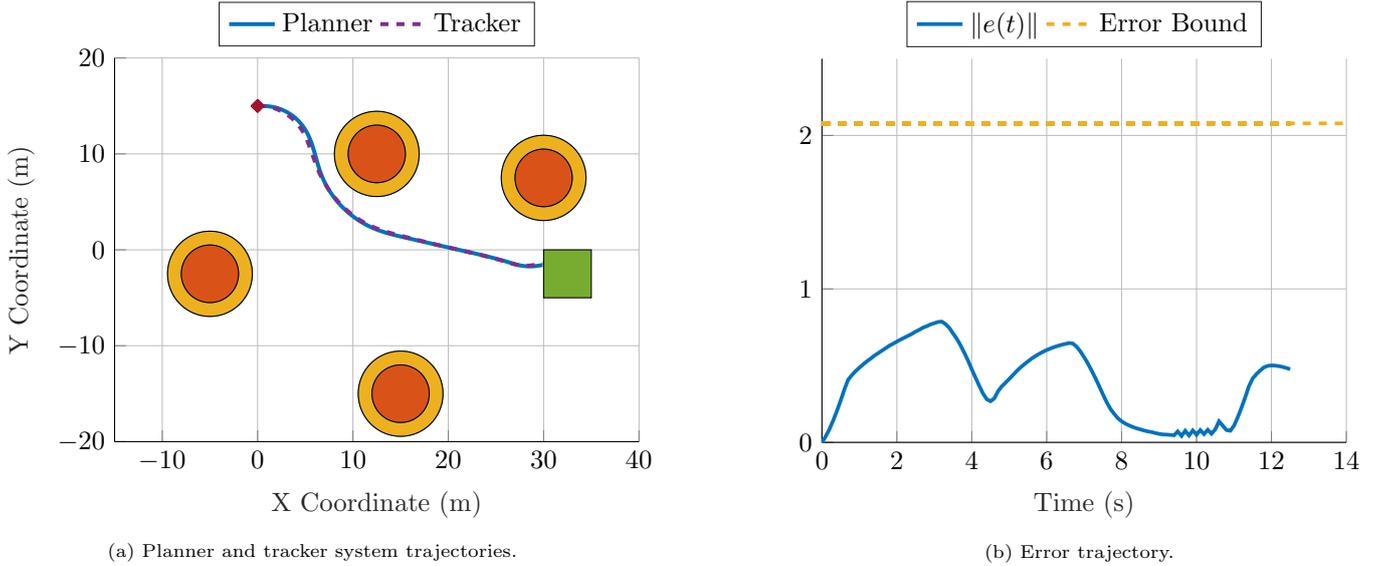

    \centering
    \begin{subfigure}{0.45\textwidth}
        \centering
        \input{plots/veh_obs_avoid_traj}
        \caption{Planner and tracker system trajectories.}
        \label{plannerTracker}
    \end{subfigure}
    \hfill
    \begin{subfigure}{0.45\textwidth}
        \centering
        \input{plots/veh_obs_avoid_err}
        \caption{Error trajectory.}
        \label{errorBound}
    \end{subfigure}
    \caption{Simulation results for the vehicle obstacle avoidance example. In Figure \ref{plannerTracker} we plot the trajectories of the planner and tracker systems through the environment, and in Figure \ref{errorBound} we plot $\|e(t)\|$ and its guaranteed upper bound. In Figure \ref{plannerTracker}, the initial position of the vehicle is marked with a red diamond and the goal set is represented with a green box. The four orange circles are the obstacles the vehicle must avoid. For each obstacle, the expanded unsafe region is shown in yellow.}
    \label{vehObsAvoidSim}
\end{figure*}

In this example, the SOS optimizations are formulated using SOSOPT~\cite{Pete:13}, and solved by MOSEK. To compute the 
tracking controller, we parameterize the storage function $V$, control law $\kappa$, and multipliers $s, l$ as degree-2 polynomials. We solve optimization~\eqref{eq:sosopt2} with these decision variables, and the computation terminates in $35.6$ seconds on a laptop with an Intel core i5 processor. Then we solve optimization~\eqref{eq:TEB_compute} with the error bound $\mathcal{O}$ parameterized as a hypercube. The resulting error bound on $(e_1, e_2, e_3)$ is $[-1.07, 1.07] \times [-1.44, 1.44] \times [-1.05, 1.05]$.

The resulting tracking controller is then tested in simulation with a corresponding planner, see Figure \ref{vehObsAvoidSim}. The objective for the planner system is to generate a pathway through the environment that avoids all obstacles and eventually reaches a goal set, which is accomplished using a standard model-predictive controller:

\begin{subequations} \label{vehObsAvoidMPC}
\begin{align}
& \underset{\hat{u}(\cdot)}{\text{min}}
& J & = \ell_f(\hat{x}(t + N_p + 1)) + \sum_{k=t}^{t+N_p} \ell(\hat{x}(k), \hat{u}(k)) \label{MPC2Obj} \\
& \text{s.t.} & & \hat{x}(k + 1) = \hat{x}(k) + T_s \cdot \hat{f}_{\text{approx}}(\hat{x}(k), \hat{u}(k)), \label{MPC2Constr1} \\
& & & \hat{x}(k) \in \hat{\mathcal{X}}, \label{MPC2Constr2} \\
& & & \hat{u}(k) \in \hat{\mathcal{U}}, \label{MPC2Constr3} \\
& & & \hat{x}(t) =\hat{x}_0 \label{MPC2Constr4}, \\
& & & \forall k = t, \dots, t+N_p, \nonumber \\
& & & \hat{u}(k) - \hat{u}(k - 1) \in \Delta \hat{\mathcal{U}}, \label{MPC2Constr5} \\
& & & \hat{u}(t) - \hat{u}_0 \in \Delta \hat{\mathcal{U}}, \label{MPC2Constr6} \\
& & & \forall k = t + 1, \dots, t+N_p,
\end{align}
\end{subequations}
{where $\ell(\cdot, \cdot)$ in (\ref{MPC2Obj}) is the state / input cost at each time step, $\ell_f(\cdot)$ is the final state cost, (\ref{MPC2Constr1}) is the polynomial approximation of the discretized Dubin's vehicle dynamics, (\ref{MPC2Constr2}) and (\ref{MPC2Constr2}) ensure state and input constraints are obeyed, and (\ref{MPC2Constr4}) is the initial state constraint.}
Furthermore, $\hat{u}_0$ is the input that was applied at the previous time step, and therefore \eqref{MPC2Constr5} and \eqref{MPC2Constr6} ensure the input jump constraints are respected. 
The objective to reach the goal set is encoded using the functions $\ell$ and $\ell_f$.
The initial state of the vehicle is $\hat{x}_0 = [0; \ 15; \ 0]$ and the goal set is a square region centered at $(32.5, -2.5)$ with a height and width of 5m. There are four circular obstacles centered at $(-5, -2.5)$, $(12.5, 10)$, $(30, 7.5)$, and $(15, -15)$, each with a radius of 3m. Since the maximum position tracking error the vehicle will experience is 1.44m, for each obstacle, we constrain the vehicle to avoid a circular region centered at the obstacle coordinates with an expanded radius of 4.44m. This ensures the vehicle will not collide with any of the obstacles. Indeed, in simulation the vehicle successfully navigates past each obstacle and eventually reaches the goal set, as shown in Figure \ref{vehObsAvoidSim}.

\section{Conclusion}\label{sec:conclusion}
In this tutorial, we address robust trajectory planning and control design for nonlinear systems. A hierarchical trajectory planning and control framework is proposed, where a
low-fidelity model is used to plan trajectories satisfying planning constraints, and a high-fidelity
model is used for synthesizing tracking controllers guaranteeing the boundedness of the error state between the low- and high-fidelity models. We consider error states that are functions of both planner states and inputs, which offers more freedom in the choice of the low-fidelity model. SOS optimizations are formulated for computing the tracking controllers and their associated tracking error bound simultaneously.
Finally, we demonstrate the planner-tracker control scheme on a vehicle obstacle avoidance example.

When implementing the planner-tracker framework in real-time, there are still challenges for providing a full guarantee of safety.
First, two sources of error in the planner dynamics are present in the example above: (1) the discretization error from the forward Euler discretization, and (2) the polynomial approximation error from the trigonometric terms. If a bound on these errors were known, it would be possible to incorporate them into the design process, ensuring instead that the planner constraints, when augmented with the discretization error, polynomial approximation error, \textit{and} the tracking error still satisfy the tracker constraints. We do not perform such an analysis in this tutorial. If the MPC problem were solved using an optimization solver that required only function evaluations of the dynamics (rather than the derivative), it would also be possible to avoid any discretization error by using a function evaluation oracle that could calculate the exact discretization via numerical integration.

In this tutorial, we also do not address the question of MPC terminal sets and costs for stability and persistent feasibility guarantees for the MPC problem. These sets/costs can be computed in simple cases but may increase the computational burden, both offline and online. Real-time reliability of solvers for MPC, especially for nonlinear models, should also be considered in practical applications.
Finally, defining the error variable can require clever selection of the function $\phi$ to make terms in the dynamics cancel, which isn't always intuitive. The shortcomings mentioned above also provide directions for further research.




\section*{Acknowledgments}
This work was supported in part by ONR Grant N00014-18-1-2209, NSF Grant ECCS-1906164, and AFOSR Grant FA9550-21-1-0288. The authors would like to thank 
Monimoy Bujarbaruah for pointing to Tube MPC references.

\bibliographystyle{abbrv}   
\bibliography{main.bib} 
\renewcommand{\baselinestretch}{1}
\section*{Appendix}
The algorithm 
to solve the bilinear optimization~(\ref{eq:sosopt1})
is summarized below, the $(\kappa,\gamma)$-step of which treats $\gamma$ as a decision variable. By minimizing $\gamma$, the volume of $\Omega(V^{j-1}, \gamma)$ can be shrunk. In the $V$-step, \eqref{eq:levelset_grow} enforces $\Omega(V^j, \gamma^j) \subseteq \Omega(V^{j-1}, \gamma^j)$.

\begin{algorithm} [h]
	\caption{Alternating direction method}
	\label{alg:alg1}
	\begin{algorithmic}[1]
		\Require{function $V^0$ such that constraints~\eqref{eq:sosopt1} are feasible by proper choice of $s, l, \kappa, \gamma$.}
		\Ensure{$\kappa, \gamma, V$.}
		\For{$j = 1:N_{\text{iter}}$}
		\State $\boldsymbol{(\kappa,\gamma)}$\textbf{-step}: decision variables $(s, l, \kappa,\gamma)$.
			
			Minimize $\gamma$ subject to \eqref{eq:sosopt1}  using $V = V^{j-1}$. 
			
			This yields ($l^j, s_{4,i}^j, s_{7,i}^j, \kappa^j$) and the cost $\gamma^j$.
		\State $\boldsymbol{V}\textbf{-step}$: decision variables $(s_{1\rightarrow3}, s_{5\rightarrow6,i},$
		
	    $s_{8\rightarrow9,i}, V)$; Maximize the feasibility subject to
	    
	    \eqref{eq:sosopt1} as well as $s_{10} - \epsilon \in \Sigma[e]$, and
			\begin{align}
			\ \ \ - s_{10} \cdot (V^{j-1} - \gamma^j) +  (V - \gamma^j) \in \Sigma[e], \ \label{eq:levelset_grow}
			\end{align}

			 using ($\gamma = \gamma^j, s_{4,i} = s_{4,i}^j, s_{7,i} = s_{7,i}^j, \kappa=\kappa^j$, 
			 
			 $l = l^j$). This yields $V^j$.
		\EndFor
	\end{algorithmic}
\end{algorithm}

The input to Algorithm~\ref{alg:alg1} is a feasible initial guess $V^0$. One candidate might be a quadratic Lyapunov function $\bar{V}$ obtained by solving Lyapunov equations using the linearized error dynamics with LQR controllers. However, $\bar{V}$ might be too coarse to be feasible for the constraints~\eqref{eq:sosopt1}. Here, we introduced a slack variable $\lambda > 0$ to the constraint~\eqref{eq:sos_cond2} to relax the constraint, and quantify how far $\bar{V}$ is away from a feasible candidate:
\begin{align}
    &-\frac{\partial V}{\partial e}\cdot(f_e+g_e \cdot \kappa) + \lambda - \epsilon e^\top e +l \cdot (V - \gamma) + s_1 \cdot p_{\hat x} \nonumber \\
    & ~~~~~~ +s_2 \cdot p_{\hat u} + s_3 \cdot p_w \in \Sigma[(e,\hat{x},\hat{u},w)]. \label{eq:relaxed_cond}
\end{align}
By iteratively search over two bilinear groups of decision variables, we minimize $\lambda$ until $\lambda \leq 0$. Based on this idea, an algorithm to compute $V^0$ from $\bar{V}$ is proposed as Algorithm~\ref{alg:alg_init}.
\begin{algorithm} [h]
	\caption{Computation of $V^0$}
	\label{alg:alg_init}
	\begin{algorithmic}[1]
		\Require{function $\bar{V}$, and $\bar{\gamma} > 0$.}
		\Ensure{$V^0$.}
		\State $V^{\text{pre}} \gets \bar{V}$
		\While{$\lambda > 0$}
		\State $\boldsymbol{\kappa}$\textbf{-step}: decision variables $(s, l, \kappa)$.
			
			Minimize $\lambda$ subject to (\ref{eq:sos_cond0}--\ref{eq:sos_cond1}, \ref{eq:relaxed_cond}, \ref{eq:sos_cond3}--\ref{eq:sos_cond4}),  
			
			using $V = V^{\text{pre}}$, $\gamma = \bar{\gamma}$. 
			
			$(l^{\text{pre}}, s_{4,i}^{\text{pre}}, s_{7,i}^{\text{pre}}, \kappa^{\text{pre}}) \gets (l, s_{4,i}, s_{7,i}, \kappa)$
		\State $\boldsymbol{V}\textbf{-step}$: decision variables $(s_{1\rightarrow3}, s_{5\rightarrow6,i},$
		
	    $s_{8\rightarrow9,i}, V)$; Minimize $\lambda$ subject to (\ref{eq:sos_cond0}--\ref{eq:sos_cond1}, \ref{eq:relaxed_cond},
	    
	    \ref{eq:sos_cond3}--\ref{eq:sos_cond4}) using ($\gamma = \bar{\gamma}, s_{4,i} = s_{4,i}^{\text{pre}}, s_{7,i} = s_{7,i}^{\text{pre}}, \kappa=$
	    
	    $\kappa^{\text{pre}}$, $l = l^{\text{pre}}$). 
	    
	    $V^{\text{pre}} \gets V$
		\EndWhile
		\State $V^0 \gets V^{\text{pre}}$
	\end{algorithmic}
\end{algorithm}


\end{document}